\documentclass[a4paper,aps,pra,showpacs,superscriptaddress,nofootinbib,reprint]{revtex4-1}
\usepackage[utf8]{inputenc}
\usepackage[T1]{fontenc}
\usepackage[UKenglish]{babel}
\usepackage{lmodern}
\usepackage[colorlinks]{hyperref}
\usepackage{graphicx}
\usepackage{amsmath,amssymb,amsthm}
\usepackage{xspace}
\usepackage{mathtools}
\usepackage{dsfont}
\usepackage{xcolor}
\usepackage{paralist}
\usepackage{array}
 
\usepackage{cancel}

\newcommand{\defeq}{\vcentcolon=}

\DeclareMathOperator{\tr}{Tr}

\newcommand{\id}{\mathds{1}}

\newcommand{\ket}[1]{\left| #1 \right\rangle}

\newcommand{\ketbra}[2]{\left|#1\middle\rangle\middle\langle#2\right|}

\newcommand{\Proj}[1]{[[#1]]}

\newcommand{\overbar}[1]{\mkern 1.5mu\overline{\mkern-2.5mu#1\mkern-1.5mu}\mkern 1.5mu}

\newtheorem{theorem}{Theorem}
\newtheorem*{theorem*}{Theorem}
\newtheorem{lemma}[theorem]{Lemma}
\newtheorem{definition}{Definition}

\newcommand{\lin}{{\cal L}}

\usepackage[normalem]{ulem} 
\usepackage{cancel} 

\begin{document}

\title{Updating the Born rule}
\author{Sally Shrapnel}
\email{s.shrapnel@uq.edu.au}
\affiliation{School of Historical and Philosophical Inquiry, The University of Queensland, St Lucia, QLD 4072, Australia}
\affiliation{Centre for Engineered Quantum Systems, School of Mathematics and Physics, The University of Queensland, St Lucia, QLD 4072, Australia}
\author{Fabio Costa}
\email{f.costa@uq.edu.au}
\affiliation{Centre for Engineered Quantum Systems, School of Mathematics and Physics, The University of Queensland, St Lucia, QLD 4072, Australia}
\author{Gerard Milburn}
\email{g.milburn@uq.du.au}
\affiliation{Centre for Engineered Quantum Systems, School of Mathematics and Physics, The University of Queensland, St Lucia, QLD 4072, Australia}

\date{\today}
\begin{abstract}
Despite the tremendous empirical success of quantum theory there is still widespread disagreement about what it can tell us about the nature of the world. A central question is whether the theory is about our \emph{knowledge} of reality, or a direct statement about reality itself. Regardless of their stance on this question, current interpretations of quantum theory regard the Born rule as fundamental and add an independent state-update (or "collapse") rule to describe how quantum states change upon measurement. In this paper we present an alternative perspective and derive a probability rule that subsumes \emph{both} the Born rule \emph{and} the collapse rule. We show that this more fundamental probability rule can provide a rigorous foundation for informational, or "knowledge-based", interpretations of quantum theory. \end{abstract}
\maketitle

\section*{}

Knowledge-based, or informational, views of quantum theory are popular for a variety of reasons. Perhaps one of the strongest motivations for this perspective comes from the conceptual difficulties that surround quantum state collapse upon measurement. If quantum states are a direct description of reality then this seems to demand that collapse is a non-linear, stochastic and temporally ill-defined physical process~\cite{ghirardi1986, Tumulka1987, Bassi2013, Gisin2017}. From a "knowledge" perspective however, collapse is seen as merely a form of information update, no more problematic than classical probabilistic conditioning \cite{Ozawa1997, Fuchs2002, Caves2002, wiseman2009quantum, timpson2013quantum, Mermin2017, Brukner2017}.

Whilst compelling, there is an obvious problem with this kind of approach: classical probabilistic conditioning treats two consecutive events on a \emph{single} system on exactly the same footing as two events on \emph{distinct} systems: joint probabilities are defined in exactly the same way in each case. In quantum mechanics however, the Born rule does not assign joint probabilities to consecutive events~\cite{Leifer2006}, Fig.~\ref{born}. This means that knowledge-based interpretations, where one argues that the Born rule is fundamental and the state-update rule "merely a case of probabilistic conditioning", are deeply unsatisfactory. Both rules have to be introduced and justified separately.

In this paper we aim to provide a solution to this problem and breathe new life into the knowledge-based view of quantum theory. We present a new, Gleason-type proof of a quantum probability rule that subsumes both the Born rule and the state-update rule. This rule is useful in a variety of contexts, from quantum information \cite{gutoski06, chiribella08, Chiribella2008, chiribella09b, Bisio2011, Bisio2014} to quantum causal modelling \cite{oreshkov12, Leifer2013, costa2016, Allen2016}, and non-markovian dynamics~\cite{modioperational2012, Ringbauer2015, pollockcomplete2015, Milz2016}. Dubbed the "Quantum Process Rule", we prove that one can derive this higher-order, generalised form of the standard quantum probability rule from the structure of quantum operations and a reasonable non-contextuality assumption. We also show that using this more fundamental approach, where one assigns joint probabilities to arbitrary quantum events, it is possible to derive both the Born rule and the state-update rule. A key advantage is that state-update, or "collapse" need no longer be viewed as an ad hoc ingredient, independent and estranged from the core of the theory.

In order to introduce the least possible assumptions, we take an explicitly operational perspective.  Operational theories can be phrased in terms of  \emph{events}, which define the results of \emph{measurements}. Each time a measurement is performed on a system, a number of possible events can be observed. The ensemble of all events that can result from a specific measurement is called a \emph{context}. 

It is natural, when constructing such a theory, to assume \emph{measurement non-contextuality} \cite{Fuchs2002, Caves2004}. This means that operationally indistinguishable events should have the same mathematical representation in the theory. Clearly, any probabilistic theory can be formulated in a non-contextual way by appropriate relabelling of the mathematical objects describing events.

In this setting, the minimal task of a physical theory is to non-contextually assign probabilities to such measurement events. In essence this is the "probability rule" of the theory and also defines the relevant state-space. One can represent any such non-contextual probability rule (the Born rule being a prime example) by means of a \emph{frame function}. This is a function that associates a probability to every event, independently of the context to which it belongs, such that probabilities for all events in a given context sum up to one. Crucially, the frame function is \emph{not} a probability distribution over the space of all events, as that would require a normalised measure over the entire space. The word "frame" here is thus synonymous with "context". 

\begin{widetext}
\onecolumngrid
\begin{figure}[ht]%
\includegraphics[width=0.499\columnwidth]{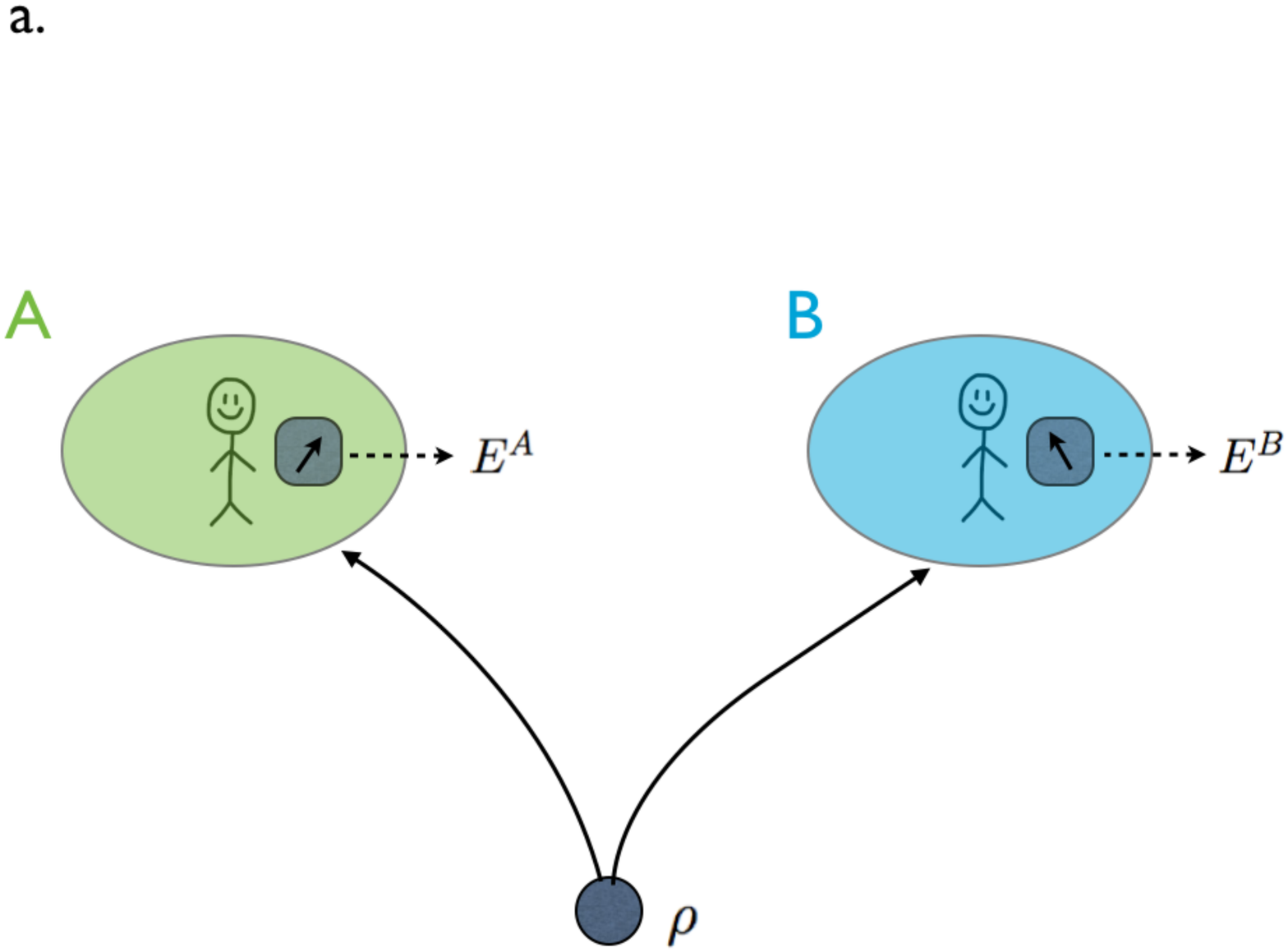} \!\!\!
\includegraphics[width=0.499\columnwidth]{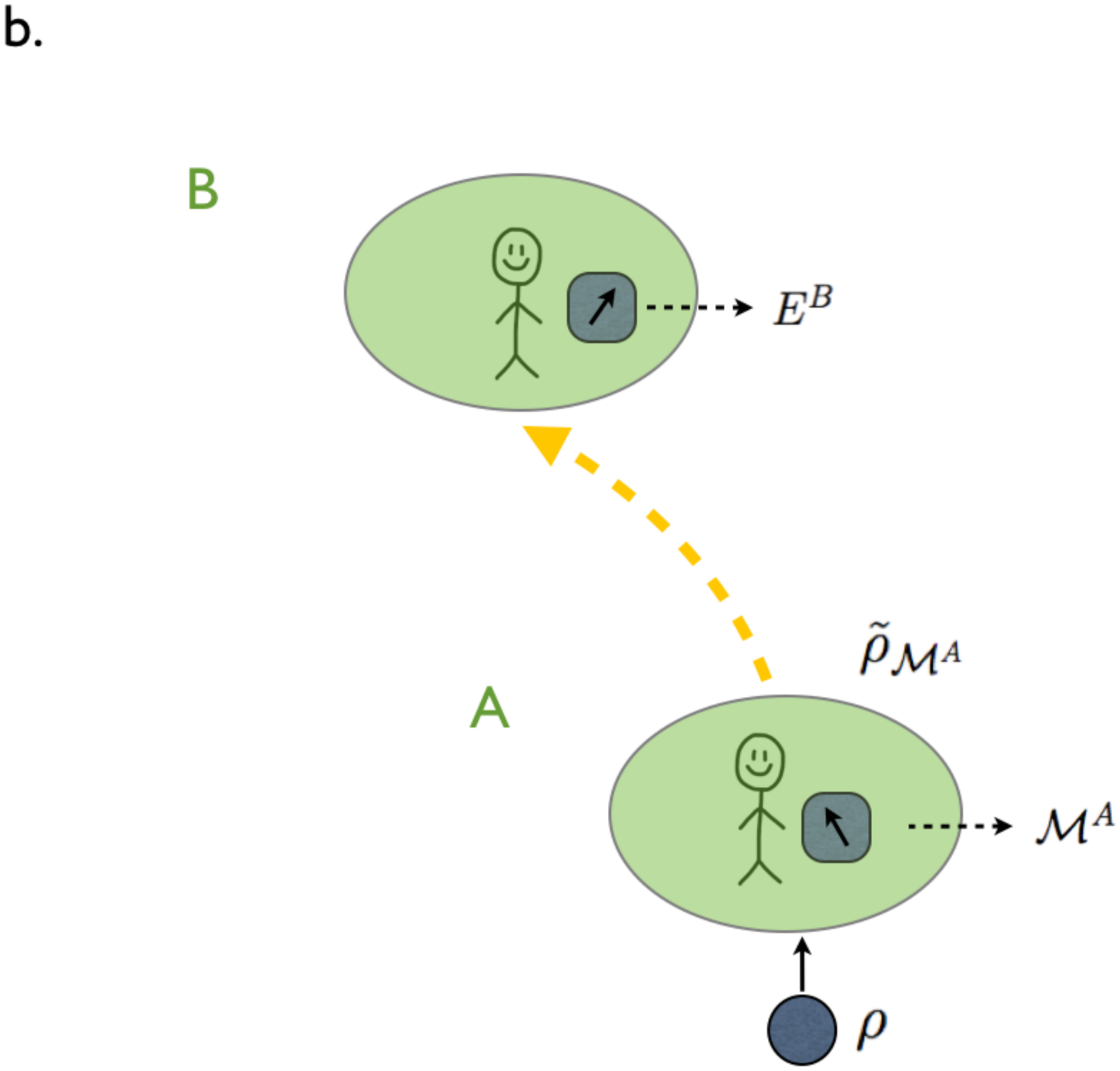} 
\caption{\textbf{Quantum probability rules}. a) The Born rule assigns probabilities to measurements on distinct systems: for a state $\rho$, and measurement operators $E^{A,B}$, the probability is $P(E^A,E^B)=\tr\left[\left(E^A\otimes E^B\right)\cdot \rho\right]$. b) For two consecutive measurements on the same system, one cannot apply the Born rule without first updating the state. The state update-rule, defined as $\rho\rightsquigarrow \widetilde{\rho}_{\mathcal{M}^A}  = \mathcal{M}^A\left(\rho\right)/\tr \mathcal{M}^A\left(\rho\right)$ for a completely positive map $\mathcal{M}^A$ describing the first measurement, is typically introduced as an independent axiom in the theory.
}%
\label{born}%
\end{figure}
\twocolumngrid
\end{widetext}

Operational approaches to quantum theory typically rely on Gleason's theorem, and generalisations thereof~\cite{gleason57, Busch2003, Caves2004, Barnett2014, hall2016comment},  to derive the Born rule. Let us briefly consider how this approach works.  Following Ref.~\cite{Caves2004}, events are identified with \emph{quantum} effects, that represent the result of a measurement on a quantum system. Formally, for a $d$-level quantum system, the full set of quantum effects is defined as $\mathcal{E}_d:=\left\{E\in \mathcal{L}\left(\mathcal{H}_d\right),\,0\leq E\leq\id\right\}$, where $\mathcal{L}\left(\mathcal{H}_d\right)$ is the space of linear operators on a $d$-dimensional Hilbert space $\mathcal{H}_d$. Contexts are described by positive-operator-valued measures (POVMs). A POVM is a complete set $X$ of effect operators that sum up to the identity, $\sum_{E\in X}{E}=\id$.

Assuming measurement non-contextuality here means that the probability of a particular quantum effect is assumed to be independent of the context (POVM) to which it belongs. Operationally, this means that the probability assigned to a given event doesn't depend on any extra information regarding how it was achieved.

A frame function for quantum effects is defined as a mapping from the set of all effects to the unit interval:
\begin{equation}
f \colon \mathcal{E}_d \rightarrow [0,1],
\end{equation}
satisfying
\begin{equation}
\sum_{E\in X} f(E) = 1 
\end{equation}
\begin{equation}
\forall X=\{E \in \mathcal{E}_d| \sum_{E\in X} E = \id \}.
\end{equation}

Using this definition, the task then is to prove that for each frame function, $f$, there is a unit-trace positive operator $\rho$ such that $f(E) = \tr (\rho E)$. 

The proof in Ref.~\cite{Caves2004} follows three simple steps. First, one proves linearity of the frame-function over the field of nonnegative rational numbers, then extension to full linearity is obtained by proving continuity of the frame-function. Then, as the frame-function has been proved to be linear, it can be recast as arising from an inner product. In particular, using the Hilbert-Schmidt inner product on the operator space $\mathcal{L}\left(\mathcal{H}\right)$, the frame-function can be written as $f(E) = \tr (\rho E)$ for some positive semidefinite, unit-trace operator $\rho$. This both characterises the Born rule and also defines the density operator as the appropriate object to represent the quantum state.

As we have noted, the above proof does not tell us how to assign probabilities to consecutive events. That is, assuming we know the state of a quantum system prior to measurement, the Born rule alone does not tell us how to update this state following measurement. To remedy this situation, we now wish to provide a similar proof for a probability rule that can subsume \emph{both} the Born rule and the state-update rule. 

We consider more general operational primitives than those of Ref.~\cite{Caves2004} and instead consider local regions where one can perform actions that are associated with outcomes. The class of allowed local actions is broad: one can perform measurements, realise transformations, or even add and discard ancillary systems. Such actions can also be associated with local outcomes and we define a particular single case outcome, associated to a given action, as the relevant \emph{event}. The event thus now labels not only the outcome but also any concurrent transformation to the local system. 

Just as with effects in the traditional approaches, we assume a minimal operational labelling for transformations: different interactions of the system with an environment, that cannot be distinguished by looking at the system alone, will be assigned the same label.

If we consider a particular run of an experiment there will in general be a collection of such events that occur, one for each local region. One can associate a joint probability to this set of events, and, given enough runs of an experiment, one can empirically verify probability assignments for each possible permutation of events.

\begin{figure}%
\includegraphics[width=0.9\columnwidth]{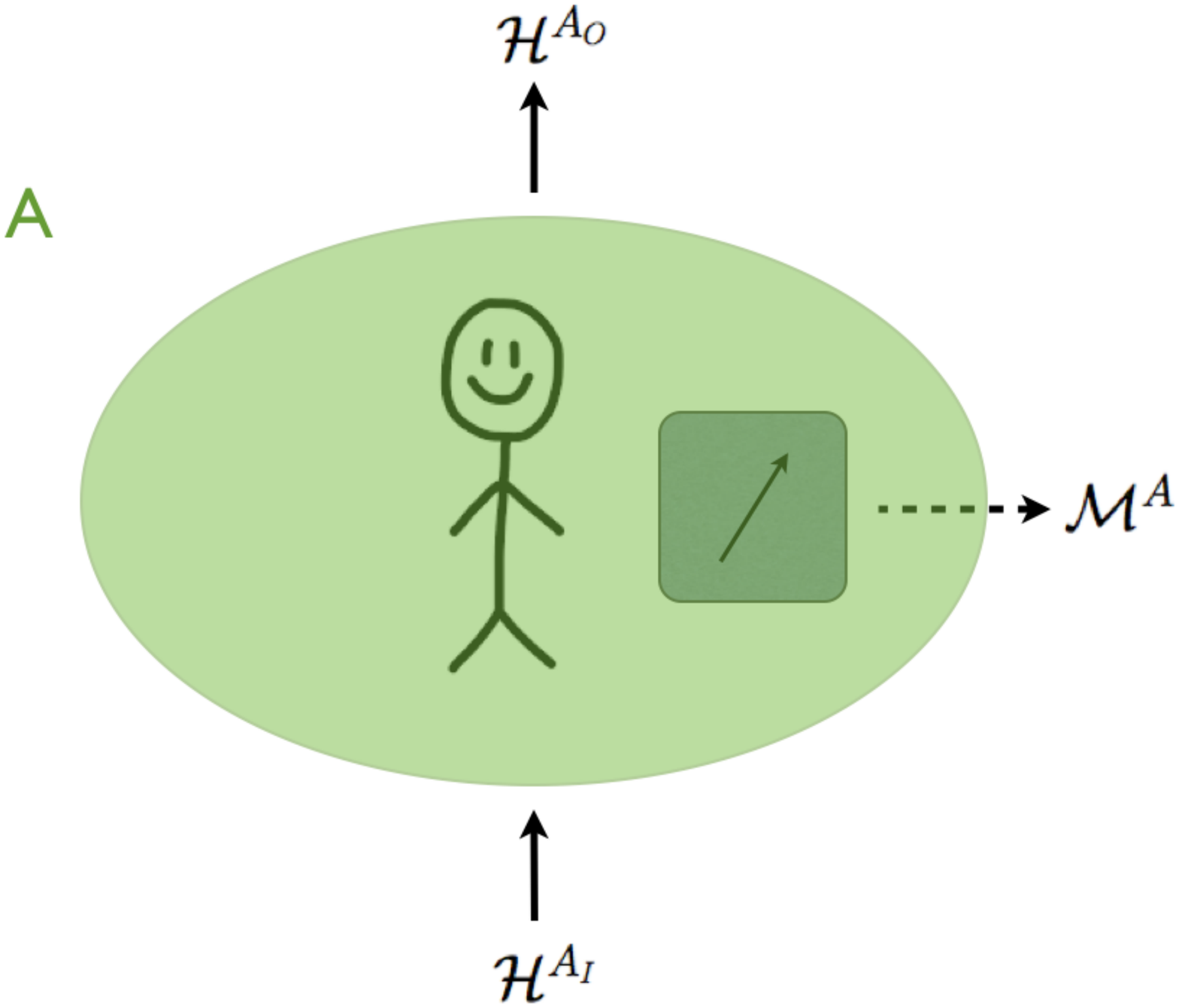}%
\caption{\textbf{Local region.} A local region $A$ is defined by an input ($\mathcal{H}^{A_I}$) and an output ($\mathcal{H}^{A_O}$) Hilbert space. An event is represented by a completely-positive map $\mathcal{M}^A$.}%
\label{region}%
\end{figure}

Formally, an event in region $A$ is represented by a \emph{completely positive trace-non-increasing} (CP) map $\mathcal{M}^A :A_I\rightarrow A_O$, where input and output spaces are the spaces of linear operators over input and output Hilbert spaces of the local region, $A_I\equiv\lin({\cal H}^{A_I})$, $A_O\equiv \lin({\cal H}^{A_O})$ respectively (here identified with the corresponding matrix spaces)~\cite{chuang00}, see Fig.~\ref{region}. We write $L^A \defeq \mathcal{L}(A_I, A_O)$ for the set of linear maps from $A_I$ to $A_O$. We denote the set of CP maps associated to each region, $CP^X \subset L^X$.

We demand \emph{complete} positivity because operationally it should be possible to perform arbitrary quantum operations in the local region. This includes performing operations on a subsystem that is part of a larger system. Complete positivity means that, for arbitrary dimensions of an ancillary system $A'$, the map ${\mathcal I}^{A'}\otimes {\mathcal{M}^A}$ transforms positive operators into positive operators, where ${\mathcal I}^{A'}$ is the identity map on ${A'}$. Trace non-increasing means that $\tr\mathcal{M}(\rho) \leq \tr\rho$ for all operators $\rho$. A CP map can be decomposed as $\mathcal{M}(\rho)=\sum_j K_j\rho K_j^{\dag}$, where the Kraus operators $K_j:{\cal H}^{A_I}\rightarrow {\cal H}^{A_O}$ satisfy $\sum_j K_j^{\dag} K_j \leq \id$ for a trace non-increasing map~\cite{Hellwig1969, Hellwig1970}

The context for each set of CP maps is now no longer a POVM but rather a \emph{quantum instrument}. An instrument thus represents the collection of all possible events that can be observed given a specific choice of local action\footnote{Note that the original definition of instrument was rather a generalisation of observable \cite{davies70}, while here we use the more recent definition as a generalisation of POVM.}. 
Given a local region $A$, an instrument is formally defined as a set ${\mathfrak I}^A$ of CP maps that sum up to a completely positive trace-preserving (CPTP) map:

\begin{equation}
\tr \sum_{\mathcal{M}^A\in {\mathfrak I}^A} \mathcal{M}^A(\rho) =  \tr(\rho).
\label{instrument}
\end{equation}

\begin{figure}[t]%
\includegraphics[width=1.1\columnwidth]{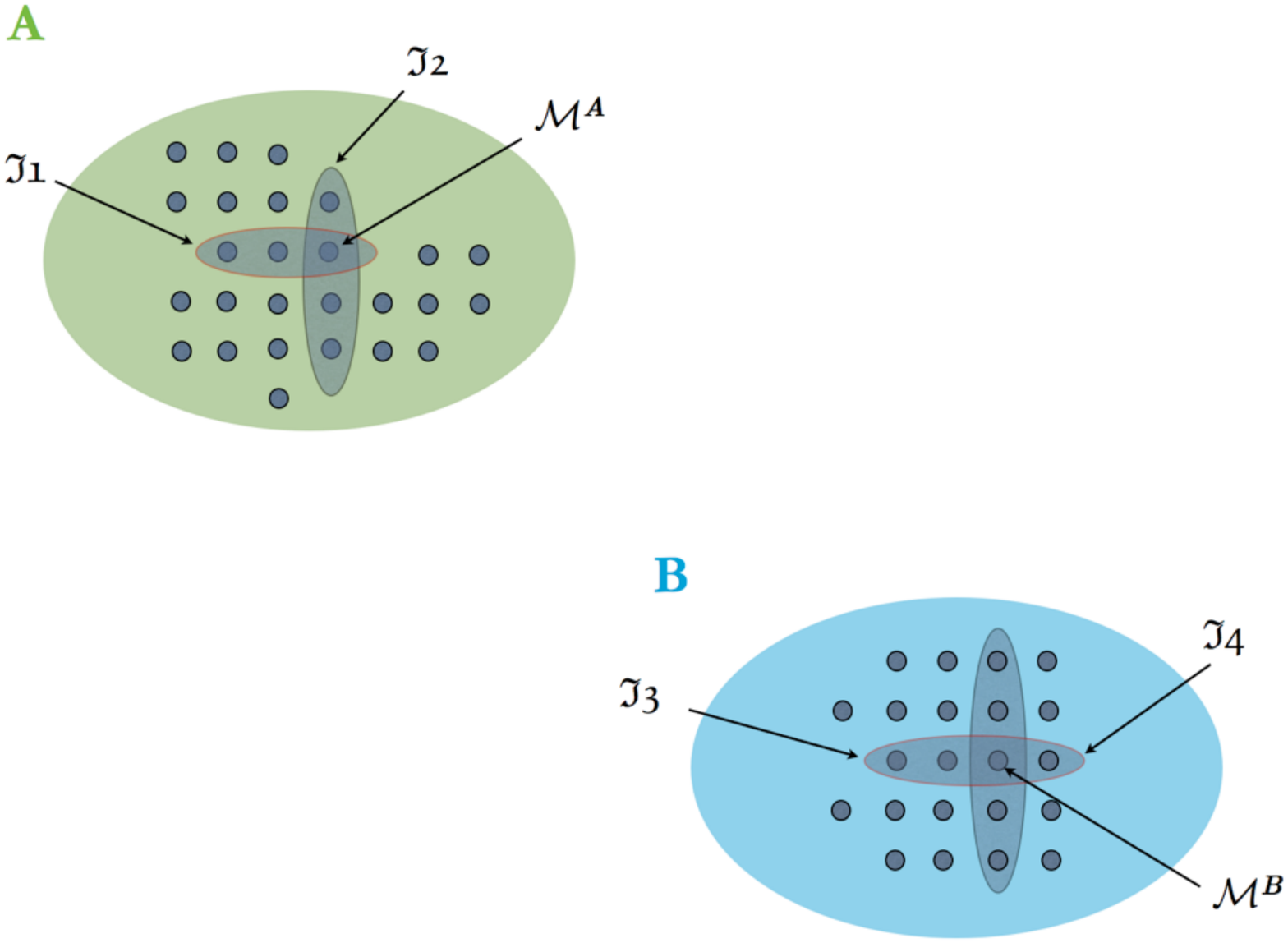}%
\caption{\textbf{Instrument non-contextuality}. Operations are performed in distinct local regions. Operation $\mathcal{M}^A$ in region $A$ corresponds to a shared outcome of two different instruments, $\mathfrak{I}_1$ and $\mathfrak{I}_2$; $\mathcal{M}^B$ in region $B$ to a shared outcome of instruments $\mathfrak{I}_3$ and $\mathfrak{I}_4$. Instrument non-contextuality implies the joint probability $P(\mathcal{M}^A, \mathcal{M}^B)$ for the two events is independent of whether instrument $\mathfrak{I}_1$ or $\mathfrak{I}_2$ was used in Region A, and whether instrument $\mathfrak{I}_3$ or $\mathfrak{I}_4$ was used in region B.}%
\label{noncontextuality}%
\end{figure}

We are now in a position to define the relevant frame-function and derive the appropriate probability rule for this scenario. Just as the Born rule tells us how to calculate the probability of a particular outcome given the relevant measurement operator, the Quantum Process Rule should tell us how to assign a joint probability to each possible collection of local events given the relevant instruments. We assume "instrument" non-contextuality, rather than "measurement" non-contextuality. That is, the joint probability for a set of events, one for each region, is independent of the particular context (set of instruments) to which they belong, see Fig.~\ref{noncontextuality}.
 
As for Ref.~\cite{Caves2004}, the non-contextuality assumption is formalised by requiring that probabilities are given by a frame-function. Each "frame" is now a collection of instruments, one per region, rather than a single POVM.

\begin{definition}\label{pframefn}
A frame-function, f, for a set of local regions {X= A, B, C....}, is defined by:

\begin{enumerate}
\item
f is a function from the cartesian product of the set of CP maps associated to each region, $CP^X \subset L^X$, to the unit interval:
\begin{equation}
f\colon CP^A \times CP^B \times CP^C ... \rightarrow [0,1]
\end{equation}
\item
f is normalised for all sets of CP maps, $\mathcal{M}^X$, that form instruments $\mathfrak{I}^X$,  

\begin {equation} 
{\sum_{\substack{\mathcal{M}^A \in \mathfrak{I}^A\\{\mathcal{M}^B \in \mathfrak{I}^B}\\{\mathcal{M}^C \in \mathfrak{I}^C}\\...}}}f(\mathcal{M}^A, \mathcal{M}^B, \mathcal{M}^C, ...) = 1
\end{equation}
\end{enumerate}

\end{definition}

We now show that this definition is sufficient to derive the new probability rule. As in Ref.~\cite{Caves2004} we first prove linearity of the frame-function. 

 \begin{theorem}\label{Wlinear}
The frame-function f is a convex-multilinear functional on $CP^A \times CP^B \times CP^C \times \dots$ 
\end{theorem} 
\noindent Where by convex-multilinear we mean:
\begin{align*}
&f\big[p \mathcal{M}_1^A + (1-p) \mathcal{M}_2^A, \mathcal{M}^B, \mathcal{M}^C, \dots \big]\\
&= pf\big[\mathcal{M}_1^A, \mathcal{M}^B, \mathcal{M}^C, \dots \big] + (1-p) f\big[\mathcal{M}_2^A, \mathcal{M}^B, \mathcal{M}^C, \dots \big]\\
&  (0\leq p \leq 1)
\end{align*}
and similarly for all other regions $B, C,\dots$

\begin{proof}
We fix instruments at all regions, except for region $A$, to be instruments with a single CPTP map each: 
$\overbar{\mathcal{M}}^B, ~\overbar {\mathcal{M}}^C,\dots$

Consider two instruments applied in region $A$:
\begin{align*}
&\mathfrak{I}^A_1 = \{\mathcal{M}^A_1, ~\mathcal{M}^A_2, ~\mathcal{M}^A_3\}\\
&\mathfrak{I}^A_2 = \{\mathcal{M}^A_1+\mathcal{M}^A_2, ~\mathcal{M}^A_3\} 
\end{align*}
The frame function constraints imply:
\begin{align*}
\!\!\!\! f\left[\mathcal{M}^A_1, ~\overbar {\mathcal{M}}^B, ...\right]  + f\left[\mathcal{M}^A_2, ~\overbar {\mathcal{M}}^B, ..\right]+f\left[\mathcal{M}^A_3, ~\overbar {\mathcal{M}}^B,..\right] &= 1\\
\!\!\!\! f\left[(\mathcal{M}^A_1+ \mathcal{M}^A_2), ~\overbar {\mathcal{M}}^B,...\right]+f\left[\mathcal{M}^A_3, ~\overbar {\mathcal{M}}^B, ...\right] &= 1
\end{align*}
Therefore 
\begin{align*}
f\left[\mathcal{M}^A_1, ~\overbar {\mathcal{M}}^B, ...\right]  + f\left[\mathcal{M}^A_2, ~\overbar {\mathcal{M}}^B, ..\right]\\
=f\left[(\mathcal{M}^A_1+ \mathcal{M}^A_2), ~\overbar {\mathcal{M}}^B,...\right]
\end{align*}
and thus we have additivity.

Separating a CP map $n\mathcal{M}^A$ into $m$ components, we can form a CP map $\frac{n}{m}\mathcal{M}^A$. Applying additivity twice

\begin{multline}
mf(\frac{n}{m} \mathcal{M}^A,\dots) = f(n\mathcal{M}^A,\dots) = nf(\mathcal{M}^A,\dots) \\
\rightarrow f(\frac{n}{m}\mathcal{M}^A,\dots)= \frac{n}{m}f(\mathcal{M}^A,\dots).
\end{multline}
Thus $f$ is linear in the nonnegative rationals.

Linearity of the frame function on the real numbers can be established using the `squeeze theorem' of elementary calculus~\cite{Smith1955}. Define two sequences of positive rationals, $\left\{a_n\right\}$ increasing and $\left\{b_n\right\}$ decreasing, that converge to the same real number $c$. Then, for any CP map $\mathcal{M}^A$, the map $\mathcal{N}_n^A:=(c-a_n)\mathcal{M}^A$ is also CP. Thus, fixing all maps in other regions to be CPTP, we have 
\begin{multline*}f\left(c \mathcal{M}^A,\dots\right) \\= f\left(a_n \mathcal{M}^A,\dots\right) + f\left(\mathcal{N}_n^A,\dots\right) \geq f\left(a_n\mathcal{M}^A,\dots\right).$$
\end{multline*}
Similarly, we have that $f\left(c \mathcal{M}^A,\dots\right) \leq f\left(b_n\mathcal{M}^A,\dots\right)$. This implies
\begin{align}
a_n f\left( \mathcal{M}^A,\dots\right) \leq f\left(c \mathcal{M}^A,\dots\right) \leq b_n f\left(\mathcal{M}^A,\dots\right).
\label{pinching}
\end{align}
Because $a_n f\left( \mathcal{M}^A,\dots\right)$ and $b_n f\left(\mathcal{M}^A,\dots\right)$ both converge to $c f\left(\mathcal{M}^A,\dots\right)$, Eq.~\eqref{pinching} implies 
\begin{equation}
f\left(c \mathcal{M}^A,\dots\right) = c f\left( \mathcal{M}^A,\dots\right)
\label{homogeneity}
\end{equation}
by the `squeeze theorem'.

We have thus proved that $f$ is linear on $CP^A$ and, with similar steps, linearity can be proven for $CP^B, CP^C, \dots$ which concludes the proof.
\end{proof}

Just as in ordinary quantum mechanics a \emph{state} is defined as a linear functional over effects (POVM elements), we can define a multilinear functional over sets of events (CP maps) as a \emph{process}, in accordance with the terminology of Refs.~\cite{oreshkov12, Brukner2014, araujo15, feixquantum2015, oreshkov15, Branciard2016, costa2016, Baumann2016, oreshkov2016, abbott2016}. 

We next use the fact that a linear functional can be expressed by means of an inner product. This enables us to derive a new probability rule using our frame function, and also gives the appropriate form for the matrix representation of a process.

First consider that because each $CP^X$ contains a basis of $L^X$, $X=A,B,\dots$, the frame function $f$ can be extended by linearity to the entire linear space $L^A \otimes L^B \otimes L^C\otimes\dots$ (as opposed to just the set of CP maps).
Next, it is easy to show that the natural inner product between any two linear maps $\mathcal{M}^A$, $\mathcal{N}^A$ $\in L^A$ is defined as follows (see Methods for details):
\begin{align}
\big(\mathcal{M}^A, \mathcal{N}^A \big) :=
 \sum_{\mu}& \tr\mathcal{M}^A(\tau_{\mu})^\dagger \mathcal{N}^A (\tau_{\mu}),
\label{CPIP} 
\end{align}
where $\left\{\tau_\mu\right\}_{\mu=0}^{d^2-1}$ is a Hilbert-Schmidt basis for the $d$-dimensional input space:
$\tau_{\mu}\in \mathcal{L}(\mathcal{H}^{A_I})$, 
$\tau_{\mu} = \tau_{\mu}^{\dagger}$, 
 $\tr \tau_{\mu}\tau_{\nu} = \delta_{\mu\nu}$.

One can also represent this inner product in a more convenient (and familiar) form by representing the CP maps associated to each region as Choi-Jamiolkowski (CJ) matrices~\cite{Choi1975, jamio72}. 
Recall, a CP map associated to a region $A$,  where input and output spaces are the spaces of linear operators over input and output Hilbert spaces, $A_I\equiv\lin({\cal H}^{A_I})$, $A_O\equiv \lin({\cal H}^{A_O})$, respectively, can be represented as a matrix\footnote{This definition aligns with the convention in Ref.~\cite{oreshkov12}. Other definitions, differing by a transpose or partial transpose, do not change the representation of the inner product}:
\begin{eqnarray}
\label{CJ}
M^{A} =& \sum_{j\,l}\ketbra{l}{j}^{A_I}\otimes \left[{\cal M}(\ketbra{j}{l})^{A_O}\right]^T ,  \label{inverseCJ}
\end{eqnarray}
where $\left\{\ket{j}\right\}_{j=1}^{d_{A_I}}$ is an orthonormal basis in ${\cal H}^{A_I}$ and $^T$ denotes transposition in that basis. We show in the Methods that the inner product \eqref{CPIP} can be expressed as
\begin{equation}
\big(\mathcal{M}^A , \mathcal {N}^A \big) = \tr M^{A\dagger}N^{A}
\end{equation}
and it is independent of the choice of Hilbert-Schmidt basis.

This inner product defines an isomorphism between elements of $L^A \otimes L^B \otimes L^C\otimes\dots$ and linear functionals on the same space. We can thus define a trace rule that allows one to determine the joint probability for a set of CP maps, one for each region:
\begin{equation} 
\begin{split} \label{Gborn}
f(\mathcal{M}^A, \mathcal{M}^B, &\dots)  \\
=&\big( \mathcal{W}_f , \mathcal{M}^A\otimes\mathcal{M}^B\otimes\dots \big) \\
=&\tr\left[\left(M^{A}\otimes M^{B} \otimes \dots \right)\cdot W_f^{A B\dots}\right],
\end{split}
\end{equation}
where $\mathcal{W}_f \in L^A \otimes L^B \otimes L^C\otimes\dots$ is the linear map that uniquely defines $f$ and $W_f^{A B\dots}$ is its CJ representation, called the  \emph{process matrix}. (In the following, we will drop the subscript $f$).

Similarly to a density matrix, the process matrix has to satisfy certain constraints so that expression \eqref{Gborn} yields a valid probability distribution for every collection of instruments. For a density matrix $\rho$, positivity of probabilities implies $\rho\geq 0$, while normalisation implies $\tr \rho =1$. For a process matrix, $W\geq 0$ is also required under the assumption that local operations can act on additional multipartite quantum states shared among the regions \cite{oreshkov12}. Normalisation imposes more complicated constraints than for density matrices; these can be expressed as linear constraints on $W$, see for example appendix B of Ref.~\cite{araujo15}.

\section*{Recovering the State update and Born rule}\label{update}
Let us recapitulate the rationale so far: it was shown in Ref~\cite{Caves2004} that if we accept the structure of quantum measurements, we can identify quantum probabilities as the most general non-contextual probability assignments. Whereas this approach only considers a single measurement/event---or at most measurements of separate quantum systems---in the quantum process approach outlined above we derive a general rule to assign joint probabilities to an arbitrary number of events. The ordinary Born rule is thus recovered from the general one in the case where a single region is considered---in which case instruments reduce to POVMs and process matrices reduce to density matrices~\cite{oreshkov12}.

We are in particular interested in the situation where two consecutive measurements are performed on a single quantum system. Ordinary Gleason-type derivations of quantum probabilities do not tell us how to assign joint probabilities to two such events: one must introduce an additional ingredient---the state update rule. If the statistics for the first measurement are described by a density matrix $\rho$, and the first measurement is described by a CP map $\mathcal{M}$, one calculates the probabilities for the second measurement, given the outcome of the first is known, by applying the Born rule to the updated state~\cite{Kraus1971}
\begin{equation}
\rho\rightsquigarrow \widetilde{\rho}_{\mathcal{M}}  = \frac{\mathcal{M}\left(\rho\right)}{\tr \mathcal{M}\left(\rho\right)} = \frac{\sum_{j}K_j\rho K_j^{\dag}}{\tr(\sum_j K_j^{\dag}K_j \rho)}.
\label{collapse}
\end{equation}
(Note that the update rule does not depend on the particular decomposition of $\mathcal{M}$ into Kraus operators $\left\{K_j\right\}_j$.)  In an operational perspective, rule \eqref{collapse} is seen as a quantum analogue of classical knowledge update. Within the quantum process framework, this is more than an analogy: the update rule is \emph{derived} from the joint probability assignment.

To make the argument rigorous, we should remark again that the quantum frame function \emph{is not} a normalised probability measure over the entire space of potential events. Formally, the frame function defines a conditional probability for observing a CP map $\mathcal{M}^A$ given an instrument $\mathfrak{I}^A$: 
\begin{equation}
\begin{matrix*}[l]
P(\mathcal{M}^A|\mathfrak{I}^A) 
&=f(\mathcal{M}^A) \quad &\textrm{if } \mathcal{M}^A\in\mathfrak{I}^A \\
&=0 \quad &\textrm{otherwise.}
\end{matrix*}
\label{frametoprob}
\end{equation}
(With a similar definition for multiple regions $A$, $B,\dots$) Even though the conditioning on the instruments is necessary to define \eqref{frametoprob} as a classical probability, we will omit it in the following out of notational convenience\footnote{That the classical probability \eqref{frametoprob} \emph{does} depend on the instrument is generally known as \emph{quantum contextuality}. This is the reason we need to introduce a frame function in the first place: it allows us to define a weaker form of noncontextuality in a theory that is, from the standpoint of classical probability theory, contextual. We remark that quantum contextuality is a general feature of quantum mechanics and not of our particular approach.}.

Expression \eqref{frametoprob} defines an ordinary, classical probability measure, which lets us use all the machinery of classical probability theory. In particular, the conditional probability to observe $\mathcal{M}^B$ in region $B$, given that $\mathcal{M}^A$ is observed in region $A$, can be calculated from the joint probability distribution:
\begin{align} \nonumber
P(\mathcal{M}^B|\mathcal{M}^A) =& \frac{P(\mathcal{M}^B,\,\mathcal{M}^A)}{P(\mathcal{M}^A)} \\ \nonumber
 =& \frac{\tr \left[\left(M^A\otimes M^B\right)\cdot W\right]}{\sum_{M^B\in \mathfrak{I}^B}\tr\left[ \left(M^A\otimes M^B\right)\cdot W\right]} \\ 
 =& \tr M^B \widetilde{W}_{M^A},
\label{bayes}
\end{align}
where we introduced the \emph{updated process matrix}
\begin{equation}
\widetilde{W}_{M^A}^{B_IB_O}:=\frac{\tr_{A_IA_O} \left[\left(M^A\otimes \id^B\right)\cdot W\right]}{\tr \left[\left(M^A\otimes \sum_{M^B\in \mathfrak{I}^B} M^B\right)\cdot W\right]}.
\label{updatew}
\end{equation}

Relevant to the ordinary state update rule is the case where $A$ precedes temporally $B$, and the evolution between the two events is trivial. This scenario is described by the process matrix (see, e.g., Ref.~\cite{costa2016})
\begin{align}
W=&\rho^{A_I}\otimes \Proj{\id}^{A_OB_I}\otimes \id^{B_O},\\
\Proj{\id}^{A_OB_I}:=& \sum_{jl}\ketbra{j}{l}^{A_O}\otimes \ketbra{j}{l}^{B_I},
\end{align}
where $\rho$ is the density matrix describing the input state of region $A$. A straightforward calculation shows that, in this case, the updated process matrix reduces to
\begin{equation}
\widetilde{W}_{M^A}^{B_IB_O} = \left[\frac{\mathcal{M}^A(\rho)}{\tr \mathcal{M}^A(\rho)}\right]^{B_I}\!\otimes \id^{B_O} \equiv \widetilde{\rho}_{\mathcal{M}^A}\otimes \id^{B_O},
\end{equation}
which is the process-matrix description of region $B$ receiving a state described by the density matrix $\widetilde{\rho}_{\mathcal{M}^A}$.

\section*{Discussion}

In this work we have shown that it is possible to use a Gleason-type approach to derive a quantum probability rule that subsumes both the Born rule and the state update rule. By using the structure of local quantum operations and a reasonable non-contextuality assumption we have derived both the new rule and the appropriate object to represent the arbitrary background structure, or process. 

Our demonstration that the state update, or "collapse" rule can be regarded as non-fundamental offers a new perspective on a variety of foundational questions. In particular, informational interpretations of wavefunction collapse can now be given a rigorous foundation: state-update can be viewed as a case of classical probabilistic conditioning.

A further advantage of the approach presented here is that it does not presuppose any a-priori distinction between space-like and time-like separated events. As such, it avoids conceptual difficulties associated with the non-covariant nature of the state update rule. It is thus a promising direction to develop a fully relativistic version of the formalism that encodes space-time symmetries.

\begin{acknowledgments}
We thank Josh Combes, Chris Timpson, and Howard Wiseman for helpful discussions. This work was supported by an Australian Research Council Centre of Excellence for Quantum Engineered Systems grant (CE 110001013), and by the Templeton World Charity Foundation (TWCF 0064/AB38). F.C.\ acknowledges support through an Australian Research Council Discovery Early Career Researcher Award (DE170100712). We acknowledge the traditional owners of the land on which the University of Queensland is situated, the Turrbal and Jagera people.
\end{acknowledgments}

\small


\section*{Methods}
\subsection*{Inner product for linear maps}
\label{innerproduct}

Here we construct the inner product on the space of linear maps $L^A=\{\mathcal{M}: \mathcal{L}(\mathcal{H}^{A_I}) \rightarrow \mathcal{L} (\mathcal{H}^{A_O})\}$ and derive its CJ representation.
Recall that, given an inner product $\left\langle \psi | \phi \right\rangle$ on a Hilbert space $\mathcal{H}$ and an arbitrary basis that is orthonormal with respect to this product, $\left\langle e_j | e_k \right\rangle=\delta_{jk}$, one defines the Hilbert-Schmidt scalar product for operators $\sigma, \rho \in \mathcal{L}(\mathcal{H})$ as
\begin{equation} \label{HS}
\big( \rho ,\sigma \big)_{\textrm{HS}}:= \sum_k  \left\langle \rho (e_k) | \sigma (e_k) \right\rangle = \tr \left(\rho^{\dag}\sigma\right),
\end{equation}
where we momentarily abandon the Dirac notation and represent explicitly the action of an operator on a vector as $v\in \mathcal{H}\rightsquigarrow \rho(v)\in \mathcal{H}$. (As it is well known, the definition of the Hilbert-Schmidt inner product does not depend on the choice of orthogonal basis.)

We  move a step further and, based on the Hilbert-Schmidt inner product, define an inner product for the space $L^A$ of linear maps. 
For this purpose, we select a basis of hermitian matrices for the input space that is orthonormal with respect to the Hilbert-Schmidt product (called Hilbert-Schmidt basis):
\begin{eqnarray*}
\tau_{\mu}\in \mathcal{L}(\mathcal{H}^{A_I}),\\
\tau_{\mu} = \tau_{\mu}^{\dagger}, \\
 \tr \tau_{\mu}\tau_{\nu} = \delta_{\mu\nu}.
\end{eqnarray*}
The inner product between any two linear maps $\mathcal{M}$, $\mathcal{N}$ is then defined in analogy to Eq.~\eqref{HS} and coincides with the inner product introduced in the main text:
\begin{align} \nonumber
\big(\mathcal{M}, \mathcal{N}\big)_{\textrm{S}} :=
\sum_{\mu}& \big(\mathcal{M}(\tau_{\mu}), \mathcal{N}(\tau_{\mu}) \big)_{\textrm{HS}} \\ 
 = \sum_{\mu}& \tr\mathcal{M}(\tau_{\mu})^\dagger \mathcal{N} (\tau_{\mu}),
\label{superinner} 
\end{align}
where the subscript S stands for ``superoperator''. Note that, just as for Eq.~\eqref{HS}, expression \eqref{superinner} formally corresponds to a trace over superoperators and is thus independent of the choice of basis.

Next, we want to relate the superoperator inner product to the CJ representation. Reintroducing the Dirac notation, the CJ inner product between operators is defined as
\begin{eqnarray} \label{CJA}
\mathcal{M} \rightarrow M^T := \sum_{jk} |j\rangle\langle k|^{A_I} \otimes\mathcal{M} ( |j\rangle\langle k|)^{A_O} , \\
\big( \mathcal{M} , \mathcal {N} \big)_{\textrm{CJ}} := \tr M^{\dagger}N.
\end{eqnarray}
Note that the inner product keeps the same form if definition \eqref{CJA} is replaced by its transpose. We can thus re-write it as
\begin{align} \nonumber
\big(\mathcal{M}&, \mathcal{N}\big)_{\textrm{CJ}}\\ \nonumber
=&\sum_{jkmn}\tr \big[|j\rangle\langle k|^{A_I} \otimes\mathcal{M} \left( |j\rangle\langle k|\right)^{A_O}\big]^\dagger |m\rangle \langle n|^{A_I} \otimes \mathcal{N}\left(|m\rangle \langle n|\right)^{A_O}\\ \nonumber
= &\sum_{jkmn}\langle j|m\rangle \langle n|k\rangle \tr \mathcal{M}(|j\rangle \langle k|)^\dagger \mathcal{N} (|m\rangle \langle n|) \\
=&\sum_{mk} \tr \mathcal{M} (|m\rangle \langle k|)^\dagger \mathcal{N}(|m\rangle \langle k|).
\label{CJdecomp}
\end{align}
To see how this relates to the superoperator inner product, we need to recall two useful facts.
\begin{lemma}
Given a Hilbert space $\mathcal{H}$, the \emph{swap operator} $S:\mathcal{H}\otimes\mathcal{H}\rightarrow \mathcal{H}\otimes\mathcal{H}$, defined by its action $S\ket{\psi}\ket{\phi}=\ket{\phi}\ket{\psi}$, can be written as 
\begin{equation}
S = \sum_{\mu} \tau_{\mu} \otimes \tau_{\mu}
\label{swapHS}
\end{equation}
 for an arbitrary Hilbert-Schmidt basis $\left\{\tau_{\mu}\right\}\subset \mathcal{L}\left(\mathcal{H}\right)$.
\end{lemma}
\begin{proof}
Viewed as an operator, $S$ can be decomposed with respect to a basis $\left\{\ket{j}\right\}$ of the Hilbert space $\mathcal{H}$ as
$S=\sum_{km}\ketbra{k}{m}\otimes\ketbra{m}{k}$. On the other hand, viewed as a vector on the linear space of operators $\mathcal{L}\left(\mathcal{H}\otimes \mathcal{H}\right)$, $S$ can be decomposed with respect to the Hilbert-Schmidt basis as 
\begin{equation}\label{HSdecomp}
S = \sum_{\mu\nu}\tau_{\mu}\otimes\tau_{\nu} \tr\left[ \left(\tau_{\mu}\otimes\tau_{\nu} \right)\cdot S\right].
\end{equation}
The components in the above representation are given by
\begin{eqnarray*}
\tr\left[ \left(\tau_{\mu}\otimes\tau_{\nu} \right)\cdot S\right] = \sum_{km} \tr\left[ \big(\tau_{\mu} \otimes \tau_{\nu}\big)\cdot \big(|k\rangle \langle m| \otimes |m \rangle \langle k|\big)\right]\\
= \sum_{km} \langle m| \tau_{\mu} |k \rangle \langle k |\tau_{\nu}|m\rangle\\= \tr \tau_{\mu} \tau_{\nu}
=\delta_{\mu\nu}.
\end{eqnarray*}
Plugging this into the decomposition \eqref{HSdecomp}, we obtain Eq.~\eqref{swapHS}.
\end{proof}

This lemma can be used to prove the completeness relation
\begin{equation}
\sum_
{\mu} \langle m|\tau_{\mu}|k\rangle^*\langle n| \tau_{\mu}|r\rangle
= \delta_{mn}\delta_{kr}.
\label{complete}
\end{equation}
Indeed, using $\tau_{\mu}=\tau_{\mu}^{\dag}$, we have
\begin{align*}
\sum_{\mu} &\langle m|\tau_{\mu}|k\rangle^*\langle n| \tau_{\mu}|r\rangle 
\\
= \sum_{\mu} &\langle k|\langle n| \tau_{\mu} \otimes \tau_{\mu} |m \rangle |r \rangle
\\ 
= \langle k|& \langle n| S |m\rangle |r\rangle 
= \delta_{mn}\delta_{kr}.
\end{align*}

We can now re-write the superoperator inner product: 
\begin{align*}
\big( \mathcal{M} , \mathcal {N} \big)_{\textrm{S}} = \sum_{\mu} \tr \mathcal{M} (\tau_{\mu})^\dagger \mathcal{N}(\tau_{\mu})\\
= \sum_{mknr} \sum_{\mu}\langle m|\tau_{\mu}|k\rangle^* \langle n |\tau_{\mu} |r\rangle \tr \mathcal{M}(|m\rangle \langle k|)^\dagger \mathcal{N} (|n\rangle \langle r|)\\
=\sum_{km} \tr \mathcal{M} \big(|m\rangle \langle k|\big)^\dagger \mathcal{N} \big(|m\rangle \langle k|\big).
\end{align*}
Comparing this with Eq.~\eqref{CJdecomp}, we conclude that
$\big(\mathcal{M}, \mathcal{N}\big)_{\textrm{CJ}}=\big(\mathcal{M}, \mathcal{N}\big)_{\textrm{S}}$.

\end{document}